\pdfoutput=1
\documentclass[letterpaper]{article}

\usepackage[margin=1in]{geometry}

\usepackage{authblk}

\usepackage{hyperref}

\title{Efficient algorithms for the Potts model on small-set expanders}

\author[1]{Charlie Carlson\thanks{CharlieAnneCarlson@ucsb.edu}}
\author[2]{Ewan Davies\thanks{research@ewandavies.org}}
\author[3]{Alexandra Kolla\thanks{akolla@ucsc.edu. Supported in part by NSF grant CCF-1452923.}}

\affil[1]{Department of Computer Science, University of California Santa Barbara, USA}
\affil[2]{Department of Computer Science, Colorado State University, Fort Collins, USA}
\affil[3]{Department of Computer Science and Engineering, University of California Santa Cruz, USA}

\usepackage{algorithm,algpseudocode}
\newcommand{\Input}{\item[{\bf Input:}]}
\newcommand{\Output}{\item[{\bf Output:}]}
\renewcommand{\Return}{\item[{\bf return}]}
\newcommand{\Goto}{{\bf goto}}
\newcommand{\e}[2]{{w(#1 \to #2)}}

\usepackage{mathtools,amssymb,amsthm}
\DeclareMathOperator{\argmax}{argmax}

\newtheorem{theorem}{Theorem}
\newtheorem{lemma}[theorem]{Lemma}
\newtheorem{claim}[theorem]{Claim}
\newtheorem{fact}[theorem]{Fact}

\newcommand*{\cC}{\mathcal{C}}

\newcommand*{\cG}{\mathcal{G}}

\newcommand*{\csproblem}[1]{\textsc{#1}}
\newcommand*{\eps}{\varepsilon}
\newcommand*{\lam}{\lambda}
\newcommand*{\Lam}{\Lambda}
\newcommand*{\gam}{\gamma}
\newcommand*{\Gam}{\Gamma}
\newcommand*{\phiin}{\phi_{\mathrm{in}}}
\newcommand*{\phiout}{\phi_{\mathrm{out}}}

\newcommand*{\Zalg}{Z_{\mathrm{alg}}}
\let\oldphi\phi
\renewcommand*{\phi}{\varphi}
\newcommand*{\ursell}{\oldphi}

\newcommand*{\sml}{\mathrm{bad}}
\newcommand*{\lrg}{\mathrm{good}}

\newcommand*{\BIS}{\#{}BIS}
\newcommand*{\SAT}{\#{}SAT}
\newcommand*{\SP}{\upshape{\textsc{Spectral Partitioning}}}

\DeclareMathOperator{\vol}{\mu}

\usepackage{cleveref}

\begin{document}

\maketitle

\begin{abstract}
	An emerging trend in approximate counting is to show that certain `low-temperature' problems are easy on typical instances, despite worst-case hardness results. 
	For the class of regular graphs one usually shows that expansion can be exploited algorithmically, and since random regular graphs are good expanders with high probability the problem is typically tractable. 
	Inspired by approaches used in subexponential-time algorithms for Unique Games, we develop an approximation algorithm for the partition function of the ferromagnetic Potts model on graphs with a small-set expansion condition. 
	In such graphs it may not suffice to explore the state space of the model close to ground states, and a novel feature of our method is to efficiently find a larger set of `pseudo-ground states' such that it is enough to explore the model around each pseudo-ground state.
\end{abstract}

\section{Introduction}

The Potts model is a probability distribution on colorings of the vertices of a graph which arises in combinatorics, approximate counting and statistical physics. 
In the ferromagnetic Potts model, which is the focus of our work, interactions between colors are chosen to favor monochromatic edges.
The model is defined by a \emph{partition function}, which turns out to be a specialization of the Tutte polynomial. 
The main algorithmic questions associated to statistical models such as the Potts model are to approximate the partition function and to sample approximately from the probability distribution. 
In probability and statistical physics one might be interested in \emph{phase transitions} as the parameters of the model vary. 
For example, on the complete graph (the \emph{mean-field} Potts model in physics terminology) the phase transition is well-understood (see e.g.~\cite{BGJ96} and the references therein). 
An intuitive description of the phase transition is that at high temperature the model is dominated by `disordered' colorings with no dominant color or long-range structure, while at low temperature the model is dominated by colorings with a dominant color which grants long-range structure to the coloring. 
The phase transition is also key to the design of approximation algorithms for the partition function. 
At high temperature the model is characterized by lack of correlation between colors of distant vertices which can be exploited algorithmically, and the description of the model at low temperatures suggests a natural approach: pick a dominant color and consider small deviations from a monochromatic coloring.
This approach was pioneered by Helmuth, Perkins and Regts~\cite{HPR19b} and has let to a growing theory of `low-temperature algorithms' for various statistical models including the ferromagnetic Potts model~\cite{JKP19,CGG+21,HJP20,BCH+20,GGS21b}. 

In the language of physics, a strategy for such low-temperature algorithms is to consider the \emph{ground states} of the model and try to efficiently enumerate the contribution to the partition function from states close to a ground state. 
For the ferromagnetic Potts model, a ground state corresponds to a coloring in which every edge is monochromatic.
In many cases there are hardness results suggesting that this approach to a low-temperature algorithm cannot work in general (unless a surprising complexity-theoretic collapse occurs), so we seek sufficient conditions on the graphs considered such that one can carry out this approach. 
A major breakthrough of Jenssen, Keevash and Perkins~\cite{JKP19} shows that the combinatorial notion of expansion suffices in several settings, including the ferromagnetic Potts model on bounded-degree graphs.
Faster algorithms based on Markov chains but using the same underlying techniques were later given by Chen, Galanis, Goldberg, Perkins and Vigoda~\cite{CGG+21}.

We give some key definitions before introducing our methods and new results.
Given a graph $G$, a number of colors $q$, and an inverse-temperature parameter $\beta>0$, the partition function of the ferromagnetic Potts model on $G$ is 
\[ Z_G(\beta) = \sum_{\omega: V(G)\to [q]} e^{\beta m_G(\omega)}, \]
where $m_G(\omega)$ counts the number of edges of $G$ which are monochromatic under the coloring $\omega$. 
Note that large $\beta$ corresponds to low temperature, giving colorings with many monochromatic edges a larger contribution to the partition function.
The main algorithmic question we study is therefore to approximate $Z_G(\beta)$ for large $\beta$.
The definition of $Z_G(\beta)$ highlights an essential entropy-energy trade-off that informs the phase transition: colorings with $m_G(\omega)$ large have larger contribution to the partition function, but only when $\beta$ is large compared to $q$ and $|E(G)|$ does this overpower the sheer number $q^{|V(G)|}$ of $q$-colorings. 

We can now identify an obstacle to extending the approach of Jenssen, Keevash and Perkins~\cite{JKP19} to a low-temperature algorithm for a broader class of graphs than expanders.   
Consider a graph that is formed from the disjoint union of two expanders (of equal size) with a small number of edges, e.g.\ a matching, added between them. 
For large $\beta$ it is not the case that all `important' states are close to ground states: if $\omega$ is such that one of the expanders is colored red and the other blue then a high proportion of the edges are monochromatic (and so $e^{\beta m_G(\omega)}$ is large), but the coloring is very far from a true ground state (a monochromatic coloring). 
Our main innovation is to overcome this obstacle and design an algorithm that works for graphs such as this near-disjoint union of two expanders. 
It turns out that this type of graph is rather natural, and emerges from a weakening of expansion known as small-set expansion that one can arrive at spectrally or combinatorially. 
We thus prove a kind of structure theorem for the Potts model on graphs with small-set expansion, stating that the graph admits a partition such that $Z_G(\beta)$ is dominated by the contribution from colorings such that each piece of the partition is near-monochromatic.

\subsection{Connections to Unique Games}

The main reason to study small-set expansion and the ferromagnetic Potts model is rather broad. 
There is a natural counting problem known as \BIS{}, which is to approximately count the number of independent sets in bipartite graphs. 
The complexity of \BIS{} is unknown: there is no known general polynomial-time approximation algorithm with small constant approximation ratio, but also no proof of hardness connecting \BIS{} to a canonical `hard' approximate counting problem such as \SAT{}.
There is an entire complexity class associated to \BIS, and a theory of approximation-preserving reductions under which we wish to know the complexity of problems in the class~\cite{DGGJ04}.
The relevance of \BIS{} to this work is that approximating the Potts partition function is \BIS{}-hard in the low-temperature regime~\cite{GJ12,GSVY16a}.

There is a superficial similarity between \BIS{} and the well-known `Unique Games Problem' from combinatorial optimization, despite appearing in somewhat distinct settings.
Unique Games is a decision problem on graphs related to \csproblem{Max-Cut}, and although it is conjectured to be NP-hard~\cite{Kho02}, no proof of this is known. 
Our approach starts with the observation that recent advances show the combinatorial notion of expansion makes both Unique Games~\cite{AKK+08,MM11} and several \BIS-hard problems~\cite{JKP19} easy in the sense of admitting a polynomial-time solution. 
In the context of Unique Games, subsequent research~\cite{Kol10,ABS15} showed that a more refined and algebraic view on expansion is highly relevant to the problem, and in particular these works give algorithms that tackle the subspace spanned by certain eigenvalues of the graph.
This led to a key discovery separating Unique Games from typical NP-hard problems: the subexponential-time algorithm for general instances of Arora, Barak and Steurer~\cite{ABS15}
The techniques of Arora, Barak and Steurer~\cite{ABS15} also highlight the importance of small-set expansion and the related spectral notion of threshold rank. 
In this work we show that small-set expansion can be exploited in approximate counting, and give an algorithm for the Potts model which requires a small-set expansion condition, or (somewhat equivalently) a suitably large gap in the spectrum of the graph. 
This generalizes previous work that relied on the usual notion of expansion~\cite{JKP19}, and our results hint at deep connections between the complexity of \BIS{} and recent advances in our understanding of the Unique Games Problem. 
In~\cite{CDK+19} it was shown that a hypothetical algorithm able to approximate suitable low-temperature Potts-like partition functions in some parameter range would refute the Unique Games Conjecture. This work goes in the other direction, attempting to put ideas that led to algorithms for Unique Games to use approximating the Potts partition function.

The main motivation for this work is the pursuit of a subexponential-time algorithm for approximating the partition function of the ferromagnetic Potts model, and thus for \BIS. 

The graphs studied in this paper represent a departure from existing trends in approximate counting. Typically, the partition functions studied are motivated by problems in physics and this guides the choice of graph instances too: lattices and random (regular) graphs are natural graphical models of physical space. 
Since random graphs are expanders with high probability, the study of expansion is motivated by these graphs---as well as by the study of the average-case complexity of approximating the partition function. 
Small-set expansion and spectral partitioning have been used successfully in the design of approximation algorithms for constraint satisfaction problems, and our work demonstrates that these techniques belong in approximate counting too.

\subsection{Preliminaries}

The partition function of the ferromagnetic Potts model on a graph $G=(V,E)$ is the function 
\[ Z_G(\beta) = \sum_{\omega : V \to [q]}e^{\beta m_G(\omega)}\,,\]
where $\beta>0$ is a parameter related to the physical notion of temperature.
The sum runs over functions $\omega$ from $V(G)$ to a set $[q]$ of $q$ colors (also known as spins), and $m_G(\omega)$ is the number of monochromatic edges of $G$ under the coloring $\omega$. 
Galanis, Štefankovič, Vigoda, and Yang~\cite{GSVY16a} showed that for $q\ge 3$ and $\beta>\beta_o(q,\Delta)$ it is \BIS-hard to approximate $Z_G(\beta)$ on graphs of maximum degree $\Delta$, where $\beta_o$ is a natural physical threshold known as the \emph{order-disorder threshold}.
In fact, they proved this for bipartite graphs though we will not restrict our attention to this class.
We note that $\beta_o\sim 2\log(q)/\Delta$ for fixed $q$ as $\Delta\to\infty$, and refer the reader to~\cite{GSVY16a} for the precise definition of $\beta_o$ and a discussion of its physical significance.

We are interested in relative approximation of real numbers, where we say that $\hat z$ is a relative $\eps$-approximation of $z$ if $e^{-\eps} \le z/\hat z\le e^\eps$.
An \emph{FPTAS} for an approximate counting problem is an algorithm that for any $\eps>0$ produces a relative $\eps$-approximation to the desired function (e.g.\ $Z_G(\beta)$ as above) in time polynomial in the size of the input and $1/\eps$. 
Our main result is an FPTAS for $Z_G(\beta)$ subject to the conditions that $G$ has maximum degree $\Delta$ and a small-set expansion condition that we discuss below, and that $\beta$ is large enough. 
This results in an efficient algorithm for a restricted version of the \BIS{}-hard approximate counting problem because we require the expansion condition on $G$, but our algorithm works for large enough $\beta$ inside the relevant parameter range.

To state our results we must first discuss expansion and related spectral concepts.
Given a graph $G=(V,E)$, the \emph{boundary} $\partial(S)$ of a set $S\subset V$ is the set of edges with exactly one endpoint in $S$. 
Similarly, the \emph{closure} $\nabla(S)$ of a set $S\subset V$ is the set of edges with at least one endpoint in $S$. 
We say that $G$ is an \emph{$\alpha$-expander} if every $S\subset V$ with $|S|\le|V|/2$ has $|\partial(S)|\ge \alpha|S|$. 
We also work with a related notion of edge expansion that is more sensitive to the \emph{volume} of a set of vertices than its size. 
Let $\vol_G(S)=\sum_{v\in S}\deg_G(v)$ be the \emph{volume} of $S$, and define the \emph{conductance} of a set $S\subset V(G)$ to be
\[ \phi_G(S) := \frac{|\partial(S)|}{\vol_G(S)}\,, \]
and the conductance of the graph $G$ itself to be
\[ \phi(G) := \min_{S : \vol(S)\le \vol(V)/2} \phi_G(S)\,. \]
We write $\vol(u)$ for the degree of the vertex $u$, since this is consistent with the definition of $\vol(\{u\})$.
Finally, we define the \emph{expansion profile}~\cite{LK99,OT12a} of $G$ to be the function given by
\[ \phi_G(\gamma) := \min_{S : \vol(S) \le \gamma} \phi_G(S)\,, \]
so that the expansion profile evaluated at $\gamma=\vol(V)/2$ is the conductance of $G$. 
The usual notion of an expander is a graph for which $\phi(G)$ is at least some constant, and a \emph{small-set expander} is a graph for which $\phi_G(\delta\mu(V))$ is at least some constant for some small $\delta\in(0,1)$. 
This definition leaves open the possibility that sets of large volume have small boundary, but enforces that sets of small volume have large boundary. 

There is a wealth of literature on the relation between the spectrum of matrices associated with $G$ and expansion of $G$, e.g.~\cite{LOT12,OT14,KLL+13}, and we let $0=\lam_1\le \lam_2\le\dotsb\le \lam_n \le 2$ be the eigenvalues of the normalized Laplacian of $G$. 
Briefly, an expander has an \emph{eigenvalue gap} between $\lam_1=0$ and $\lam_2$, and a natural expansion condition is simply $\lam_2 \ge \Omega(1)$. 
Our results essentially require a more general condition which corresponds to an eigenvalue gap between $\lam_{k-1}$ and $\lam_k$ for some constant $k$.

\subsection{Our results}

Recall that the primary motivation in this work is to develop for the \BIS{}-hard problem of approximating $Z_G(\beta)$ results that parallel the advances in algorithms for the Unique Games Problem. Accordingly, we state results that require small-set expansion conditions which generalize the $\lam_2 \ge \Omega(1)$ definition of expansion.

\begin{theorem}\label{thm:regularsse}
    Let $G=(V,E)$ be a $\Delta$-regular graph on $n$ vertices.
    There is an absolute constant $C$ such that the following holds. 
    Suppose that for some integer $k\ge 2$, we have the small-set expansion condition that
	\[ |\partial(S)|\ge C\Delta k^6\sqrt{\lam_{k-1}}|S| \]
	for all sets $S\subset V(G)$ of size at most $n/k$. 
    Then for the $q$-color ferromagnetic Potts model with 
	\begin{equation*}
		\beta \ge Ck^6\cdot \frac{4 + 2\log(q\Delta)}{\lam_k^2\Delta}\,,
	\end{equation*}
    there is a deterministic algorithm that produces a relative $\eps$-approximation to $Z_G(\beta)$ in time polynomial in $n$, $1/\eps$ and $2^k$ provided $q$ and $\Delta$ are constant. For constant $k$ this gives an FPTAS\@.
\end{theorem}

\noindent
For $k=2$ the small-set expansion condition is trivially true, and the above result is essentially the same as an algorithm in~\cite{JKP19}. 
Perhaps the weakest natural small-set expansion condition for larger $k$ is a generalization of $\alpha$-expansion that would say for some constant $k\ge 2$ that sets $S$ of size at most $n/k$ have $|\partial(S)|\ge\alpha|S|$. 
The condition above is a slight strengthening of this that requires $\alpha$ to be at least some function of $k$, $\lam_{k-1}$, $\Delta$. 

The above result is an immediate corollary of the following more general form that dispenses with the regularity assumption.

\begin{theorem}\label{thm:sse}
    Let $G=(V,E)$ be a graph on $n$ vertices with maximum degree $\Delta$ and minimum degree $\delta$. 
    There is an absolute constant $C$ such that the following holds. 
    Suppose that for some integer $k\ge 2$,
    \[ \phi_G\left(\frac{\Delta}{\delta}\frac{\vol(V)}{k}\right) \ge Ck^6\sqrt{\lam_{k-1}}\,. \] 
    Then for the $q$-color ferromagnetic Potts model with 
	\begin{equation}\label{eq:sse:beta}
		\beta \ge Ck^6\cdot \frac{4 + 2\log(q\Delta)}{\lam_k^2\delta}\,,
	\end{equation}
    there is a deterministic algorithm that produces a relative $\eps$-approximation to $Z_G(\beta)$ in time polynomial in $n$, $1/\eps$ and $2^k$ provided $q$ and $\Delta$ are constant. For constant $k$ this gives an FPTAS\@.
\end{theorem}

We can obtain improved running times at the expense of requiring randomness in these approximation algorithms by swapping our main tool, the cluster expansion, for a Markov chain. 
The details of this idea are given in~\cite{CGG+21}, and how they apply to the similar algorithms presented in~\cite{JKP19}. 
It is routine to verify that the conditions required for the methods of~\cite{CGG+21} are implied by our work here, and so we immediately obtain more efficient, randomized versions of our algorithms with those methods. 
We do not discuss this in detail here.

Our methods for approximating $Z_G(\beta)$ require a partition of $V(G)$ such that (i) each part induces an expander with large minimum degree, (ii) few edges lie between parts, and (iii) no part is too small. 
Our main algorithmic result exploits a partition with a guarantee on the size of the smallest part. 
The following technical result is used to prove both \cref{thm:regularsse,thm:sse}.
We say that a partition $P_1,\dotsc,P_\ell$ of $V(G)$ is a $(\phiin,\phiout,\tau\delta)$-partition if for all $i\in[\ell]$ we have $\phi(G[P_i])\ge\phiin$, $\phi_G(P_i)\le\phiout$ and $G[P_i]$ has minimum degree at least $\tau\delta$.

\begin{theorem}\label{thm:nobadparts}
	Suppose that we have a graph $G=(V,E)$ on $n$ vertices of maximum degree $\Delta$.
	Suppose also that for some $\ell\ge 1$ we have sets $P_1,\dotsc,P_\ell$ that form a $(\phiin,\phiout,\tau\delta)$-partition of $V$ such that $|P_i|\ge \eta n$ for all $i\in[\ell]$. 
	Then for the $q$-color ferromagnetic Potts model on $G$ with 
	\[ \beta \ge \frac{2+4\log(q\Delta)}{\phiin\tau\delta\eta} \,,\]
	there is a deterministic algorithm that produces a relative $\xi$-approximation to $Z_G(\beta)$ in time at most
	\[ O\big(q^\ell\Delta n (2n/\xi)^{O(\log(q\Delta))} \big)\,.\]
\end{theorem}

\noindent
Note that we do not make use of $\phiout$ in the conclusion of the theorem. 
While (a function of) $\phiout$ gives an upper bound on the number of edges leaving each part, so does the minimum degree condition inside each part and this is the bottleneck for our methods.

To derive \cref{thm:regularsse,thm:sse} from this result we develop a mild strengthening of a spectral partitioning result of Oveis Gharan and Trevisan~\cite{OT14} that yields a partition of a graph into expanders with control on the number of edges between the parts. 
The following theorem is an extension of the main algorithmic result in~\cite{OT14} which, subject to an eigenvalue gap $\lam_k \gg \lam_{k-1}$, gives a partition of the type required for our results. In the case that $G$ has minimum degree $\delta$, we add the minimum degree condition parametrized by $\tau\delta$ to the partition with the guarantee that $\tau\ge \Omega(1/k)$.

\begin{theorem}\label{thm:partitioning}
    Let $G=(V,E)$ be a graph on $n$ vertices of minimum degree $\delta$.
    Given $k\ge 2$ such that $\lam_k>0$ there is an algorithm that runs in time $O(k n^2|E|^2)$ which yields a partition $P_1,\dotsc,P_\ell$ of $V$ into $\ell<k$ pieces that forms a $(\phiin,\phiout,\tau\delta)$-partition where
    \[ \phiin \ge \Omega(\lam_k^2/k^4)\,,\qquad \phiout \le O\left(k^6\sqrt{\lam_{k-1}}\right)\,,\qquad \tau\ge\Omega(1/k)\,. \]
\end{theorem}

\noindent
In fact, we can ensure the mildly stronger property that for all $i\in[\ell]$ and for all $v\in P_i$ we have $\vol_{G[P_i]}(v)\ge\tau\vol_G(v)$, but the definitions given above are more convenient in our application.

One of the main reasons to be interested in expansion is that random graphs typically have excellent expansion properties. 
But any random graph model with strong anisotropy could fail to have expansion, while still having readily exploitable structure.
To further motivate our work, we show how \cref{thm:nobadparts} applies to the regular stochastic block model, or \emph{RSBM}, as defined in e.g.~\cite{KP20} and motivated therein as a natural model of a clustered network. 

The RSBM is defined by positive integers $d$, $k$, $n$, and a $k\times k$ symmetric matrix of strictly positive integers $A$ whose row sums are all equal to $d$ (and often whose diagonal entries are at least $3$, meaning $d\ge 3$ also).
Then $n$ vertices are divided into $k$ equal-sized communities $P_1,\dotsc,P_k$ where for each $1\le i\le k$ the community $P_i$ is a random $A_{i,i}$-regular graph and for each $1\le i<j\le k$ we put a random $A_{i,j}$-regular bipartite graph on $(P_i,P_j)$. 
The resulting graph $G$ is $d$-regular, but the distribution need not be close to the random $d$-regular graph, for example when the diagonal entries dominate the off-diagonal entries in $A$ the graph resembles $k$ random graphs which are loosely connected. 
This gives a natural model with $k$ equal-sized communities that form connections inside their communities much more readily than between communities. We use the regular variant of the stochastic block model for convenience, a similar result holds for the usual definition of the model.

\begin{theorem}\label{thm:rsbm}
	Let $d,k,n$ be positive integers such that $n$ is a multiple of $k$ and $dn$ is even. Suppose that for fixed $\eps\in[0,1/2)$, $d$ is sufficiently large and we have a $k\times k$ symmetric matrix of strictly positive integers $A$ whose row sums are all equal to $d$ such that the diagonal entries are all at least $(1-\eps)d$. 
	Let $G$ be an instance of the regular stochastic block model defined by $d,k,n,A$. 
	Then, given the community identities, for any fixed $\zeta>0$, with probability $1-o(1)$ as $n\to \infty$, for all 
	\[ \beta \ge 16k\frac{1+2\log(qd)}{d} \]
	there is an algorithm that produces a relative $\xi$-approximation to $Z_G(\beta)$ in time 
	\[ O(q^kdn(2n/\xi)^{O(qd)}). \]
	For $k,q,d$ constant this gives an FPTAS\@.
\end{theorem}

Although this algorithm requires the community identities to proceed, spectral properties of the matrix $A$ can be used to efficiently, approximately recover the identities in the sense of \emph{weak recovery}, see~\cite[Theorem 4.2]{KP20}.
We avoid getting into the details here, but the algorithm of~\cite{KP20} together with \cref{thm:rsbm} implies that there is a natural spectral condition on $A$ one can add to \cref{thm:rsbm} that obviates the need to know community identities, and still gives an FPTAS when $k,q,d$ are constant.

Note that the statement of \cref{thm:sse} specialized to regular graphs is precisely \cref{thm:regularsse}, so to prove both these theorems we can focus on \cref{thm:sse}.
The expansion profile assumption means that with $\eta=1/k$, every part in the partition guaranteed by \cref{thm:partitioning} must have size at least $\eta n$, as we now justify. 
The partition guarantee in Theorem~\ref{thm:partitioning} states that the resulting $P_1,\dotsc,P_\ell$ has $\phi(P_i)\le \phiout = O(k^6\sqrt{\lam_{k-1}})$. 
Suppose that $\phiout >0$. 
Now an assumption $\phi(\Delta\vol(V)/(k\delta)) \ge 2\phiout$ means that any $S\subset V$ with $|S|\le n/k$, and hence $\vol(S)\le \Delta\vol(V)/(k\delta)$, has $\phi(S)\ge 2\phiout$. 
Thus, every $P_i$ has $|P_i|\ge \eta n$ with $\eta=1/k$. 
Theorem~\ref{thm:sse} now follows from Theorems~\ref{thm:partitioning} and~\ref{thm:nobadparts}.

Proving \cref{thm:rsbm} is a simple matter of using the properties of random regular graphs. 
Using Cheeger's inequality in the form $\phi(G)\ge \lam_2/2$ and a result of Friedman~\cite{Fri08}, we know that the random $d$-regular graph $G$ has with probability $1-o(1)$ as the number of vertices tends to infinity,
\[ \phi(G) \ge \frac{1}{2} - \frac{\sqrt{d-1}}{d} - \frac{\eps}{d}. \]
Then the assumptions of \cref{thm:rsbm} give that with probability $1-o(1)$, the partition $P_1,\dotsc,P_k$ given by the communities, which are of size $n/k$, has $\phi_{G[P_i]} \ge (1-\eps)/2$ for each $i$ (where we use that $d$ is sufficiently large). 
Each $G[P_i]$ also has minimum degree at least $(1-\eps)d$, which allows a direct application of \cref{thm:nobadparts}.
These simple derivations mean that it now suffices to prove \cref{thm:nobadparts,thm:partitioning}.

\subsection{Technical overview}
Our main proof technique is based on the natural idea from physics that systems at low temperatures are typically characterized by their ground states and small deviations from them. 
In our setting, a ground state is a coloring $\omega\in [q]^{V(G)}$ which maximizes the number $m_G(\omega)$ of monochromatic edges, and so has a maximum contribution to the partition function $Z_G(\beta)$ over all possible colorings.
A deviation from a fixed ground state is represented by a set of connected induced subgraph with colorings that differ from the ground state coloring.
We will show that one only has to consider deviations of a bounded size using a standard tool known as the \emph{cluster expansion} of a carefully constructed \emph{abstract polymer model}.
An abstract polymer model is an auxiliary physical model that considers assignments to \emph{polymers} in the graph instead of just vertices.
We will use connected induced subgraphs as our polymers.
The cluster expansion of this model is an infinite series for the logarithm of the partition function of this model.
We appeal to a general convergence criterion~\cite{KP86} for this series and evaluate a truncation of the series to approximate the partition function. 
The remaining task is then to show that the sum over colorings that constitutes our function $Z_G(\beta)$ of interest can be broken up into pieces that are (approximately) disjoint and well-approximated by carefully chosen abstract polymer models. 

This approach has its roots in statistical physics, and the fact that the cluster expansion yields approximation algorithms was first explored in~\cite{HPR19b}. Amongst several notable subsequent applications, the method was used to great effect in~\cite{JKP19,CGG+21} where an FPTAS and an FPRAS (a randomized version of an FPTAS) for $Z_G(\beta)$ and related partition functions were given for expander graphs. 
Our work is an extension of the relevant ideas in the following way. 
Some ground states for $Z_G(\beta)$ are trivial to compute and work with, as they are simply the colorings that give every vertex the same color. 
But when the graph has low conductance (i.e.\ is nearly disconnected) there may be non-monochromatic colorings that are very close to being a ground state, and very far from the actual ground states. 
This presents a genuine obstruction to the approach of~\cite{JKP19}, which we overcome here.
Our main idea is to partition the sum over colorings in $Z_G(\beta)$ into contributions from states that we call \emph{pseudo-ground states} and devise a polymer model for each one, rather that relying only on true ground states. 
The main task one must solve to apply our method is therefore a decomposition of the desired state space based on pseudo-ground states that can be handled with the cluster expansion (or related tools such as the Markov chains in~\cite{CGG+21}). 
Here, we show that spectral properties of the input graph can be exploited efficiently to give a decomposition into parts such that the set of pseudo-ground states given by colorings that make each part monochromatic, but allow different parts to be different colors, suffices to approximate $Z_G(\beta)$. 
To achieve this we use tools from spectral graph theory developed for other purposes in theoretical computer science~\cite{OT14}, refined for our purposes of approximating partition functions.

\subsection{Related work}

Galanis, Goldberg and Stewart~\cite{GGS21b} have generalized the approach of~\cite{JKP19} in a different direction, giving faster algorithms for general spin systems on regular bipartite expanders.
Our work shows that a sophisticated understanding of the spectrum of the underlying graph can be exploited in approximate counting problems such as the Potts model. 
We remark that spectral considerations are used extensively in approximate counting as Markov chains are important tools in the area, and bounding their mixing time can be done via spectral properties of the transition matrix. 
Our work does not rely on Markov chains and so is rather different from work of this kind, a notable exception being a result of Alev and Lau~\cite{AL20} which gives an algorithm for approximately counting the number of independent sets of a given size which exploits the absolute value of the smallest eigenvalue of a graph. 

In a very recent development, Jenssen, Perkins, and Potukuchi~\cite{JPP21a} significantly improved the algorithmic approach of~\cite{JKP19} for counting independent sets in bipartite graphs and weakened the necessary expansion condition. 
As an interesting application, they gave an algorithm for \BIS{} on regular graphs (with a polynomial approximation ratio) that runs in subexponential time when the degree of the graph grows with the number of vertices. 
This result is substantial evidence in support of a general subexponential-time algorithm for \BIS.

So far we have discussed our work in the context of several recent advances in algorithms using the cluster expansion from statistical physics, but here we give a broader picture of the literature.
Recently, the cluster expansion has been used to give an approximation algorithm for the ferromagnetic Potts model on integer lattice graphs at all temperatures (i.e.\ for all $\beta>0$)~\cite{BCH+20}. 
Since the present work first appeared as a preprint, Helmuth, Jenssen, and Perkins~\cite{HJP20} studied a generalization of the Potts model known as the \emph{random cluster model} and gave a detailed picture of the model on the class of bounded-degree graphs with both an expansion and a small-set expansion condition. 
They focus on the condition that for $\delta\in(0,1/2)$ and $\Delta\ge 3$, $G$ is of maximum degree $\Delta$ and the expansion profile satisfies $\phi_G(1/2)\ge 1/10$ and $\phi_G(\delta)\ge 5/9$ (though the precise constants can essentially be anything $\Omega(1)$ and $1/2+\Omega(1)$ respectively). 
Amongst several results characterizing the physical properties of the model, they give an FPTAS for $Z_G(\beta)$ when $G$ satisfies the above condition, for all $\beta>0$ when $q$ is large enough compared to $\delta$ and $\Delta$.
This result is not directly comparable to our \Cref{thm:sse} because we dispense with the rather strong assumption $\phi_G(1/2)\ge\Omega(1)$, which means our algorithm applies to a very different class of graphs. 
In short, for reasons related to work on Unique Games we wish to exploit a gap in spectrum of the graph wherever it may appear, whereas in~\cite{HJP20} the authors study an expansion condition met by the random regular graph. 
The random regular graph is a natural object to develop algorithms for, typically studied in physics due to its connections to the \emph{Bethe lattice} (also known as the infinite regular tree). 

\subsection{Further questions}

Our work hints at several further problems in approximate counting. 
It would be interesting to adapt our methods to counting independent sets in bipartite graphs directly, or rather approximation of the partition function of the hard-core model at low temperatures on bipartite graphs. 
The authors of~\cite{JKP19} show how to do this for expanders, but generalizing the methods to small-set expanders and the output of our Theorem~\ref{thm:partitioning} is a natural next step.

Our methods also present a challenge to spectral partitioning techniques for irregular graphs. On regular graphs the notions of $\alpha$-expansion and conductance coincide nicely, but some loss related to such a translation in irregular graphs is unfortunate in this work. 
The spectral partitioning techniques that proceed via the normalized Laplacian naturally feature volume and conductance instead of set size in the relevant places, and adapting these methods to fit together better is certainly desirable for our applications.

It would be very interesting to extend our techniques and attempt to give a general subexponential-time algorithm for a \BIS-hard problem, along the lines of the aforementioned work for Unique Games. 
Our techniques are a step towards this goal that show the spectrum and partitioning into expanders can be exploited in approximate counting. 
The main limitation of our method is that we require some properties of the partition that are hard to ensure in general: that there are no small parts and that the minimum degree inside parts is large. 
We suspect that a general method for handling `problematic' edges between parts in approximate counting problems can overcome the first of these limitations, and hope that our work motivates further research on this topic. 
In the case of Unique Games, `problematic' edges can be omitted without catastrophic effects. A more sophisticated approach seems to be necessary for the Potts model and other approximate counting problems.

\subsection{Organisation}

In Section~\ref{sec:cluster} we present an overview of our method, then give the details and show how Theorems~\ref{thm:regularsse} and~\ref{thm:sse} follow from the method and Theorem~\ref{thm:partitioning}.
In Section~\ref{sec:partitioning} we prove Theorem~\ref{thm:partitioning}, which completes the proofs of our main results.

\section{Approximating Potts partition functions}\label{sec:cluster}

In this section we prove \cref{thm:nobadparts}, and we choose to remove the restriction on part sizes by removing any edges that leave small parts. 
That is, we develop the following generalization of \cref{thm:nobadparts}.

\begin{theorem}\label{thm:Zalg}
    Suppose that we have a graph $G=(V,E)$ on $n$ vertices, of maximum degree $\Delta$ and of minimum degree $\delta$.
    Suppose also that for some $\ell\ge 1$ we have sets $P_1,\dotsc,P_\ell$ that form a $(\phiin,\phiout,\tau\delta)$-partition of $V$ such that for some $0\le s\le \ell$ we have $|P_i| < \eta n$ for $1\le i \le s$, and $|P_i|\ge \eta n$ for $s+1\le i\le \ell$. 
    Then for $\xi>0$ and 
    \[ \beta \ge \frac{4+2\log(q\Delta)}{\phiin\tau\delta\eta} \,,\]
    there is a deterministic algorithm giving a relative $((s+1)\xi+\beta s\phiout\Delta\eta n/2)$-approximation to $Z_G(\beta)$ in time at most
    \[ O\left(\left(s\eta^{O(\log(q\Delta))} + q^{\ell-s}\right)\Delta n\left(\frac{2n}{\xi}\right)^{O(\log(q\Delta))}\right) \,.\]
\end{theorem}

We call the sets $P_1,\dotsc,P_s$ which each contain less than $\eta n$ vertices \emph{bad parts}, and the remaining sets $P_{s+1},\dotsc,P_\ell$ \emph{good parts}. 
The details of why a small part is `bad' emerge in the proof, but it essentially arises because of a tension in the needs for (a) the partition function to be dominated by the contribution from colorings close to a pseudo-ground state, and for (b) a method of approximating this contribution for each pseudo-ground state. 
If one relaxes the definition of `close' then one needs stronger approximation methods as there are more states that must be captured as close to a given pseudo-ground state. 
Here our pseudo-ground states are those in which each part is monochromatic, and any state in which each part has a majority color is close to the unique pseudo-ground state which uses only that majority color on each part. 
The majority color definition simultaneously avoids any annoyances arising from colorings being close to multiple pseudo-ground states, and more importantly allows us to exploit the expansion assumption in each part of the partition. 
It is unclear how to improve these methods to capture, e.g.\ a coloring which splits the colors red, blue, and green equally amongst vertices of the smallest part, colors all the other parts red.

The proof proceeds as follows. 
For $1\le i\le s$, we compute a relative $\xi$-approximation $\hat Z_i$ to $Z_{G[P_i]}(\beta)$. 
Since $G[P_i]$ is an expander in the sense that $\phi(G[P_i])\ge \phiin$, we can use the Potts model result of~\cite{JKP19} directly to achieve this efficiently.

Now let $G'$ be the subgraph of $G$ graph induced by the good parts.
We adapt the method of~\cite{JKP19} to work with $G'$ which is supplied with a $(\phiin, \phiout, \tau\delta)$-partition $P_{s+1},\dotsc,P_\ell$, and such that each part has size at least $\eta n$. 
That is, we define a polymer model to represent the contribution to $Z_{G'}(\beta)$ from states close to a pseudo-ground state in which each $P_i$ is monochromatic. 
We use the cluster expansion to approximate the partition function of each polymer model, and show that an appropriate sum of these yields a relative $\xi$-approximation $\hat Z'$ to $Z_{G'}(\beta)$.

Putting the pieces together, we have an approximate partition function for the subgraph $G''$ of $G$ obtained by removing from $G$ the edges $\partial(P_i)$ for $1\le i\le s$. 
Since we have a $\phiout$ guarantee for the partition $P_1,\dotsc,P_\ell$ we can bound the number of such edges, and then use
\[ \hat Z := \hat Z' \cdot \prod_{1\le i\le s}\hat Z_i \]
as an approximation of $Z_G(\beta)$ that is good enough to prove Theorem~\ref{thm:Zalg}.

The partition $P_1,\dotsc,P_\ell$ that we assume for Theorem~\ref{thm:Zalg} has the guarantee that $\phi(G[P_i])\ge \phiin$. 
It is convenient to restate this in terms of a boundary condition akin to the one given in~\cite{JKP19}.
Recall that we say the graph $G$ is an \emph{$\alpha$-expander} if $|\partial(S)|\ge \alpha|S|$ for all $S$ with $|S|\le |V(G)|/2$, and note the following fact to relate $\alpha$-expansion and conductance.

\begin{fact}\label{fact:expander}
    Let $G=(V,E)$ be a graph of minimum degree $\delta$.
    Then $G$ is an $\alpha$-expander for $\alpha=\delta\phi(G)$.
\end{fact}
\begin{proof}
Let $|S|\le |V|/2$, and note that this means $|S|\le|V\setminus S|$. 
If $\vol(S)\le \vol(V)/2$ then we have $\phi(S)\ge \phi(G)$, and so 
\[ |\partial(S)| = \phi(S)\vol(S) \ge \phi(G)\delta|S|\,. \]
If instead we have $\vol(S) > \vol(V)/2$, then $\vol(V\setminus S)\le\vol(V)/2$. This means
\[ |\partial(S)| = |\partial(V\setminus S)| \ge \phi(G)\delta|V\setminus S|\ge \phi(G)\delta|S|\,. \qedhere \]
\end{proof}

In the proof of Theorem~\ref{thm:Zalg} we can use this fact to show that each $G[P_i]$ is an $\alpha$-expander with $\alpha=\phiin\tau\delta$, because $G[P_i]$ has $\phi(G[P_i])\ge\phiin$ and minimum degree at least $\tau\delta$.

Our main tool (in addition to the partitioning result Theorem~\ref{thm:partitioning}) is the cluster expansion for abstract polymer models. 
An abstract polymer model is a collection of polymers $\gam$, a weight $w_\gam$ for each polymer, and a compatibility relation on polymers. 
To be concrete, when approximating the Potts partition function of a graph $G$ our polymers are (special) connected subsets of vertices of $G$, polymers are compatible if they are at distance greater than 1 in $G$, and the weight is some natural weighted sum over colorings of the polymer.
The partition function of an abstract polymer model is $\Xi := \sum_\Gamma\prod_{\gamma\in\Gamma}w_\gamma$, where the sum is over all finite sets $\Gam$ of mutually compatible polymers. 
The cluster expansion is the formal series in the weights given by
\[ \log \Xi = \sum_{\Gamma'}\ursell(H(\Gamma'))\prod_{\gamma\in\Gamma'}w_\gamma\,, \]
where this time the sum is over \emph{clusters} $\Gam'$ which are ordered, finite lists of polymers whose graph formed by the incompatibility relation is connected, and the term $\ursell(H(\Gamma'))$ is a combinatorial term depending on the compatibility graph induced by the cluster.
There are standard convergence criteria that guarantee this series is convergent and does so rapidly, and our algorithms use these criteria to show that a suitably truncated cluster expansion gives the desired relative approximation to $\Xi$. 
See e.g.~\cite{JKP19} or~\cite{HPR19a} for more details on the use of the cluster expansion to design algorithms.
We will not require detailed knowledge of the terms of the cluster expansion, it suffices to apply existing packaging of the necessary ideas from e.g.~\cite{JKP19}.

\subsection{Dealing with the bad parts}

One of the main results of~\cite{JKP19} is the following approximation algorithm for the low-temperature ferromagnetic Potts model on $\alpha$-expander graphs of bounded degree.

\begin{theorem}[Jenssen, Keevash, and Perkins~\cite{JKP19}]\label{thm:JKP}
	Let $G$ be an $\alpha$-expander on $n$ vertices of maximum degree $\Delta$. 
	Then for the $q$-color ferromagnetic Potts model on $G$ with 
	\[ \beta\ge \frac{4+2\log(q\Delta)}{\alpha}\,, \]
	there is a deterministic algorithm that produces a relative $\xi$-approximation to the partition function $Z_G(\beta)$ in time at most $O\big(\Delta n(2n/\xi)^{O(\log(q\Delta))}\big)$.
\end{theorem}

For $1\le i\le s$, each $G[P_i]$ is a $\phiin\tau\delta$-expander on less than $\eta n$ vertices, and has maximum degree at most $\Delta$.
Via the above result we can obtain a relative $\xi$-approximation $\hat Z_i$ to $Z_{G[P_i]}(\beta)$ when 
\[ \beta\ge \frac{4+2\log(q\Delta)}{\phiin\tau\delta}\,, \]
and to do this for $1\le i\le s$ takes total time at most 
\[ T_\sml = O\big(s \Delta \eta n(2\eta n/\xi)^{O(\log(q\Delta))}\big)\,.\]

\subsection{Dealing with the good parts}

Here we develop a generalization of the methods of~\cite{JKP19} that were used to prove Theorem~\ref{thm:JKP}. 
In the course of the proof we will see that the minimum size of a part in the partition affects the range of $\beta$ for which our argument works. 
To avoid defining separate notation for a graph and partition with no bad parts, in this subsection we work with some graph $G$ on $n$ vertices with a partition $P_1,\dotsc,P_\ell$ that has no bad parts. 
Then when we apply the result below as a step in the proof of Theorem~\ref{thm:Zalg}, we will apply it with $G=G'$, $\ell=\ell'-s$, some $n'\le n$, and other changes in notation. 
That means this section yields a direct proof of \cref{thm:nobadparts}.

We say that a set $U\subset V$ is \emph{small} if for all $i\in[\ell]$ we have $|U\cap P_i| \le |P_i|/2$. 
Let $\alpha := \phiin\tau\delta$, so that by Fact~\ref{fact:expander} we have that each $G[P_i]$ is an $\alpha$-expander.
A small set has a convenient `partition expansion' guarantee because the condition $|U\cap P_i|\le|P_i|/2$ implies that for all $i$ we have $|\partial_{G[P_i]}(U\cap P_i)|\ge\alpha$, where we use a subscript on the $\partial$ to denote the graph in which we take the boundary. 
For brevity, we will write $U_i := U\cap P_i$, and $\partial_i$ for $\partial_{G[P_i]}$.  

We say that a set $U\subset V$ is \emph{sparse} if each of the connected components of $G[U]$ is small, noting that any small set is necessarily sparse.
When $\Gam$ is the set of connected components of $G[U]$, for all $i\in[\ell]$ we have,
\[ \partial_i(U_i) = \bigcup_{\gam\in\Gam}\partial_i(\gam_i)\,, \]
where $\gam_i = \gamma \cap P_i$ and the union is disjoint since there are no edges of $G$ between the sets $\gam$.
Then if $U$ is sparse we can count for each $i$ the edges leaving $U$ inside $G[P_i]$,
\[
|\partial_i(U_i)| = \sum_{\gam\in\Gam}|\partial_i(\gam_i)| \ge \alpha\sum_{\gam\in\Gam}|\gam_i| = \alpha|U_i|\,,
\]
showing that sparse sets also have a convenient `partition expansion' property.

Recall that the $q$-color ferromagnetic Potts model on $G$ is given by the partition function
\[
Z_{G}(\beta) = \sum_{\omega : [q]^V} e^{\beta m_{G}(\omega)}\,,
\]
where $m_{G}(\omega)$ counts the number of monochromatic edges of $G$ under the coloring $\omega$. 
We refer to the colorings $\omega$ as states, and note that the states giving all vertices of $G$ the same color have the highest possible contribution to the partition function, $e^{\beta |E|}$. 
A state with this maximum contribution is usually called a \emph{ground state}, but we are interested in a slightly more flexible notion of ground state.
We consider any coloring such that each $P_i$ is monochromatic a ground state, and note that when there are few edges crossing the partition these states all contribute to the partition function a term close to the maximum.
We hope the reader will allow this abuse of terminology, as writing pseudo-ground throughout this section seems overly verbose.

Our goal is to show that $Z_{G}(\beta)$ is well-approximated by the contribution from states close to a ground state, and to show that we can efficiently approximate each such contribution with the cluster expansion. 
Both of these steps inevitably require large $\beta$. 
It is simply false that the partition function is dominated by states close to a ground state for small $\beta$, and we will use large $\beta$ when appealing to a standard condition that guarantees our series approximation via the cluster expansion is convergent. 

Let $\Psi$ be the set of ground states, noting that $|\Psi|= q^{\ell}$. 
When $\psi:V\to[q]$ is a ground state, i.e.\ when $|\psi(P_i)|=1$ for all $i\in[\ell]$, we write $\psi_i$ for the unique color such that $\psi(P_i)=\{\psi_i\}$.
We say that a state $\omega$ is \emph{close} to $\psi$ if, for all $i\in[\ell]$ we have $|\omega^{-1}(\psi_i)\cap P_i| > |P_i|/2$, and note that each $\omega$ is close to at most one ground state.
We write $\Omega:=[q]^V$ for the set of all states, $\Omega^\psi$ for the set of states close to $\psi$, and $\Omega^*:=\bigcup_{\psi\in\Psi}\Omega^\psi$.

\subsubsection{Approximation by states close to a ground state}

Let $Z^*_{G}(\beta) := \sum_{\omega\in\Omega^*}e^{\beta m_{G}(\omega)}$ be the contribution to $Z_{G}(\beta)$ from states close to any ground state, and for a specific ground state $\psi$ let $Z^\psi_{G}(\beta) := \sum_{\omega\in\Omega^\psi}e^{\beta m_{G}(\omega)}$ be the contribution from states close to $\psi$.

We show that $Z^*_{G}(\beta)$ is close to $Z_G(\beta)$ for large enough $\beta$.
If $\omega\in\Omega\setminus\Omega^*$ then there is some $i\in[\ell]$ such that for all colors $c\in[q]$ we have $\omega^{-1}(c)_i \le |P_i|/2$.
Note that for every color $c$, $\partial_i(\omega^{-1}(c)_i)$ consists entirely of bichromatic edges, and hence $P_i$ contains at least
\begin{align*}
\frac{1}{2}\sum_{c\in[q]}|\partial_i(\omega^{-1}(c)_i)| 
\ge \frac{1}{2}\sum_{c\in[q]}\alpha|\omega^{-1}(c)_i|
= \frac{1}{2}\alpha |P_i|
\ge \frac{\alpha\eta n}{2}
\end{align*}
bichromatic edges.
Then $e^{\beta m_{G}(\omega)} \le e^{\beta(|E|-\alpha\eta n)}$. This means
\[
Z_{G}(\beta) - Z_{G}^*(\beta) = \sum_{\omega\in\Omega\setminus\Omega^*}e^{\beta m_{G}(\omega)} \le q^n e^{\beta(|E|-\alpha\eta n)}\,.
\]
Then the fact that $Z_{G}(\beta) \ge qe^{\beta|E|}$ means that
\[
0 \le 1-\frac{Z^*_{G}(\beta)}{Z_{G}(\beta)} \le q^{n-1}e^{-\beta \alpha\eta n}\,.
\]
When $\beta \ge 2\log(eq)/(\alpha\eta)$ we have
\[
e^{e^{-n}} > 1 \ge \frac{Z^*_{G}(\beta)}{Z_{G}(\beta)} \ge 1-e^{-n}/q \ge e^{-e^{-n}}\,,
\]
and hence $Z^*_G(\beta)$ is a relative $e^{-n}$-approximation to $Z_{G}(\beta)$ (for the last inequality we need $q\ge 2$ and $n\ge 1$).
Note that this step is making crucial use of the minimum part size guarantee. 
If the minimum part size was $o(n)$ (instead of $\eta n$), then the lower bound on $\beta$ would tend to infinity with $n$, which is undesirable.

\subsubsection{Approximation by a polymer model}

Since $Z^*_{G}(\beta)=\sum_{\psi\in\Psi}Z^\psi_{G}(\beta)$ is a good approximation to $Z_{G}(\beta)$, it will suffice to obtain good approximations to each $Z^\psi_{G}(\beta)$ and take their sum.
We do this with a polymer model and cluster expansion for each $\psi$.

Let a \emph{polymer} be a set $\gam\subset V$ such that $G[\gam]$ is connected and $\gam$ is small, $|\gamma \cap P_i| \le |P_i|/2$ for all $i$.
Two polymers $\gam$ and $\gam'$ are \emph{compatible} if they are disjoint sets of vertices such that $\partial(\gam) \cap \partial(\gam') = \emptyset$.
We write $\cC$ for the set of polymers and $\cG$ for the family of sets of mutually compatible polymers.

Now fix a ground state $\psi$.
For a set $U\subset V$, let $\Omega(U,\psi)$ be the set of states $\omega\in\Omega$ such that for all $v\in U$ we have $\omega(v)\ne\psi(v)$ and for all $v\in V\setminus U$ we have $\omega(v)=\psi(v)$.
That is, for $\omega\in \Omega(U,\psi)$ the set $U$ encodes the places where $\omega$ and $\psi$ differ.
We also write $\Lam(U,\psi)$ for the set of colorings $\lam:U\to[q]$ such that for all $v\in U$ we have $\lambda(v)\ne\psi(v)$.
We change letters for this definition to signify that $\lam$ is a coloring of $U$ alone, and not all of $V$.

Then for $\lam\in\Lam(U,\psi)$ let $m_{G}(\psi,U,\lam)$ be the number of edges in $G$ which have at least one endpoint in $U$ that are monochromatic when $U$ is colored by $\lam$ and any vertex at graph distance exactly one from $U$ is colored by $\psi$.
We write
\[
R^\psi(U,\beta) := \sum_{\lam\in\Lam(U,\psi)} e^{\beta m_{G}(\psi,U,\lam)}
\]
for the \emph{restricted partition function of $U$ with boundary conditions specified by $\psi$}.
Note that when $\Gam$ is the set of connected components of $G[U]$ we have
\[
e^{-\beta|\nabla(U)|}R^\psi(U,\beta) = \prod_{\gam\in\Gam}e^{-\beta|\nabla(\gam)|}R^\psi(\gam,\beta)\,,
\]
by the fact that there are no edges between the sets in $\Gam$. 
This lack of edges between the $\gam\in\Gam$ also permits the simultaneous imposition of boundary conditions for each $\gam$ without conflict. 
The key point of these definitions is that with weights
\[
w_\gam := e^{-\beta|\nabla(\gam)|}R^\psi(\gam,\beta)\,,
\]
and the polymer model partition function $\Xi^\psi :=  \sum_{\Gam\in\cG}\prod_{\gam\in\Gam}w_\gam$ we will be able to show that
\[
\tilde Z^\psi_{G}(\beta) := e^{\beta m_{G}(\psi)} \Xi^\psi
\]
is a good approximation of $Z^\psi_G(\beta)$.

Fix an arbitrary subset $U\subset V$, a state $\omega\in\Omega(U,\psi)$, and let $\Gam$ be the set of connected components of $G[U]$.
Then every edge that does not intersect $U$ is colored the same under $\omega$ and $\psi$, so
\begin{equation}\label{eq:mGexpr}
	\begin{split}
m_{G}(\omega) &= m_{G}(\psi) - |\nabla(U)| + m_{G}(\psi,U,\omega|_U) 
\\&= m_G(\psi) - \sum_{\gam\in\Gam}|\nabla(\gam)| + \sum_{\gam\in\Gam}m_{G}(\psi, \gam, \omega|_\gam)\,.
	\end{split}
\end{equation}
In the case that $U$ is sparse, each component of $U$ is small and hence forms a polymer.
Then there is a one-to-one correspondence between sparse subsets $U\subset V$ and $\Gam\in\cG$ (given by $U=\bigcup_{\gam\in\Gam}\gam$) and so
\[
\tilde Z^\psi_{G}(\beta) = e^{\beta m_{G}(\psi)}\sum_{\Gam\in\cG}\prod_{\gam\in\Gam}e^{-\beta|\nabla(\gam)|}R^\psi(\gam,\beta) = \sum_{U\text{ sparse}}\sum_{\omega\in\Omega(U,\psi)}e^{\beta m_{G}(\omega)}\,.
\]
But we also have
\[
Z^\psi_{G}(\beta) = \sum_{U\text{ small}}\sum_{\omega\in\Omega(U,\psi)}e^{\beta m_{G}(\omega)}\,,
\]
because $\omega$ being close to $\psi$ means that the set $U$ where $\omega$ and $\psi$ differ satisfies $|U_i|\le |P_i|/2$ for all $i\in[\ell]$, i.e.\ $U$ is small.
Recall that a small set is sparse, which means
\begin{equation}\label{eq:clusterdiff}
\tilde Z^\psi_{G}(\beta)-Z^\psi_{G}(\beta) = \sum_{\substack{U\text{ sparse,}\\\text{not small}}}\; \sum_{\omega\in\Omega(U,\psi)}e^{\beta m_{G}(\omega)}\,.
\end{equation}

Hence, if $U$ is sparse we have for each index $i$ that $|\partial_i(U_i)|\ge \alpha|U_i|$, and hence for all pairs $(U,\omega)$ appearing in the double-sum in~\eqref{eq:clusterdiff}, and for and any index $i\in[\ell]$, there are at least $\alpha|U_i|$ bichromatic edges inside $P_i$ under $\omega$.
But every $U$ in the sum in~\eqref{eq:clusterdiff} is not small and hence there is some index $i\in[\ell]$ such that $|U_i| > |P_i|/2$. 
This means there are more than
\[ \alpha|P_i|/2 \ge \alpha\eta n /2 \]
bichromatic edges in $\nabla(U)$ under $\omega$.
That is, for all pairs $(U,\omega)$ appearing in the double-sum in~\eqref{eq:clusterdiff}, 
\[
m_{G}(\psi,U,\omega|_U) < |\nabla(U)| - \alpha\eta n/2\,,
\]
and hence via~\eqref{eq:mGexpr} we have $m_{G}(\omega) < m_{G}(\psi) - \alpha\eta n/2$.
Returning to~\eqref{eq:clusterdiff} this means
\[
0\le \tilde Z^\psi_{G}(\beta)-Z^\psi_{G}(\beta) \le q^{n} e^{\beta\big(m_{G}(\psi) - \alpha\eta n/2\big)}\,,
\]
because the number of $\omega$ that are in $\Omega(U,\psi)$ for some sparse, not small $U$ is (crudely) at most $q^{n}$.
We also have that $Z^\psi_{G}(\beta) \ge e^{\beta m_{G}(\psi)}$ so that
\[
1 \le \frac{\tilde Z^\psi_G(\beta)}{Z^\psi_G(\beta)} \le 1 + q^{n} e^{-\beta\alpha\eta n/2} \le 1+ e^{-n} \le  e^{e^{-n}}
\]
when $\beta\ge 2\log(eq)/(\alpha\eta)$.

\subsubsection{Convergence of the cluster expansion}
We seek to apply the following theorem stated by Jenssen, Keevash, and Perkins~\cite{JKP19}, with many variations used elsewhere, e.g.,~\cite{BCH+20} and the earlier~\cite{HPR19a}. 
The theorem concerns a surprisingly useful special case of abstract polymer models that appears in statistical physics and counting algorithms. The setting is that we have an ambient graph $G$, polymers are vertex subsets of $G$ which induce connected graphs, and the polymer compatibility relation is given by graph distance in $G$. Our application of abstract polymer models fits this special case, which is sometimes called a \emph{subset gas} or \emph{subset polymer model}.

\begin{theorem}\label{thm:JKPconvergence}
	Suppose that $G=(V,E)$ is a graph on $n$ vertices of maximum degree $\Delta$, and that the following hold for a polymer model associated to $G$ and some decay function $g(\cdot)$ on polymers, such that polymers are subsets of $V$ that induce connected subgraphs and polymers are compatible when they are at graph distance at least two.
	
	\begin{enumerate}
		\item Given set $\gam\subset V$ with $G[\gam]$ connected, determining whether $\gam$ is a polymer and computing $w_\gam$ can be done in time $e^{O(|\gam|)}$,
		\item there exists $\rho > 0$ so that for every polymer $\gam$, $g(\gam)\ge \rho|\gam|$, and
		\item the Kotecký--Preiss condition holds with the given function $g(\cdot)$ in the sense that for all polymers $\gam$,
		\[ \sum_{\gam' : d_G(\gam,\gam')\le 1} |w_{\gam'}|e^{|\gam'| + g(\gam')} \le |\gam| \,.\]
	\end{enumerate}
	Then there is a deterministic algorithm that gives a relative $\xi$-approximation to $\Xi$ in time 
	\[ O\Big( n \cdot (n/\eps)^{O(\log\Delta /\rho)}\Big)\,. \]
\end{theorem}

The first condition is straightforward for our polymer models. 
Given a set $\gam\subset V$ that induces a connected subgraph, we have to count the size of each $\gam_i$ to decide whether $\gam$ is a polymer. 
This can be done in time $O(|\gam|)$. 
To compute the weight $w_\gam$ we sum over the $(q-1)^{|\gam|}$ colorings of $\gam$ that appear in the sum giving $R^\psi(\gam,\beta)$ and count monochromatic edges.
This takes time at most $(q-1)^\Delta \cdot \Delta|\gam| = e^{O(\Delta|\gam|\log q)}$.
We will take $g(\gam)=|\gam|$ so the second condition is satisfied with $\rho=1$.
For the final condition we exploit the fact that every polymer $\gam$ is small and hence expands in each $G[P_i]$, and that every edge in $\partial_i(\gam_i)$ is bichromatic for any coloring $\lam$ considered in the sum giving $R^\psi(\gam,\beta)$.
Then
\[
w_\gam \le (q-1)^{|\gam|}e^{-\beta\alpha|\gam|}\,,
\]
so it suffices to show that
\begin{equation}\label{eq:KPcalc}
\sum_{\gam' : d(\gam,\gam')\le 1} e^{(2-\beta\alpha+\log(q-1))|\gam'|} \le |\gam|\,.
\end{equation}
We would be done if we could prove that for each $v\in V$ we have
\[
\sum_{\gam'\in\cC \,:\, v\in\gam'}e^{(2-\beta\alpha+\log(q-1))|\gam'|} \le  \frac{1}{\Delta+1}\,,
\]
because by summing this inequality over all vertices $v$ at distance at most $1$ from $\gam$, of which there are at most $(\Delta+1)|\gam|$, we have~\eqref{eq:KPcalc}.
For this argument we use that fact~\cite[Lemma 2.1]{GK04} that there are at most $(e\Delta)^t$ connected, induced subgraphs of $G$ on $t$ vertices that contain any fixed $v\in V$.
Then
\[
\sum_{\gam'\in\cC : v\in\gam'}e^{(2-\beta\alpha+\log(q-1))|\gam'|} \le \sum_{t=1}^\infty e^{(3-\beta\alpha+\log(q-1)+ \log \Delta)t} = \frac{e^a}{1-e^a}\,,
\]
where $a := 3-\beta\alpha+\log(q-1)+ \log \Delta$. 
Then the above sum is at most $1/(\Delta+1)$ when $a \le -\log(\Delta+2)$, which holds e.g.\
when $\beta \ge (4+2\log(q\Delta))/\alpha$.
For such $\beta$ we have via Theorem~\ref{thm:JKPconvergence} an FPTAS for $\tilde Z^\psi_{G}(\beta)$ that yields a relative $\xi$-approximation $\hat Z^\psi_{G}$ to $\tilde Z^\psi_G(\beta)$ in time 
\[ O\left( n \cdot (n/\xi)^{O(\log \Delta)}\right)\,. \]

\subsubsection{Finishing the argument for good parts}

In the pursuit of Theorem~\ref{thm:nobadparts}, we have shown that when 
\[ \beta \ge \frac{2\log(eq)}{\alpha\eta}\,, \] 
the function $Z^*_{G}(\beta)$ is a relative $e^{-n}$-approximation of $Z_{G}(\beta)$.
Recall that $Z^*_{G}(\beta)$ is the sum of $q^{\ell}$ terms $Z^\psi_{G}(\beta)$, and for the same condition on $\beta$ we have for each $\psi\in\Psi$ the function $\tilde Z^\psi_G(\beta)$ which is a relative $e^{-n}$-approximation of $Z^\psi_G(\beta)$.
Finally, when 
\[
\beta \ge \frac{4+2\log(q\Delta)}{\alpha}\,,
\] 
we have an algorithm to approximate each $\tilde Z^\psi_G(\beta)$.

We now turn to the algorithm for approximating $Z_{G}(\beta)$ when 
\[ \beta \ge \frac{4+2\log(q\Delta)}{\eta\alpha}\,,\] 
as in Theorem~\ref{thm:nobadparts}, and so $\beta$ satisfies both of these conditions. 
We want a relative $\xi$-approximation, and if $\xi \le e^{-n/2}$ then $1/\xi \ge e^{n/2}$ and so we have time to compute via brute force. 
We can compute $Z_{G}(\beta)$ in time $O(q^n \cdot \Delta n) = O((1/\xi)^{2\log q}\cdot \Delta n)$, which a polynomial in $1/\eps$ and $n$.
Otherwise, for $\zeta=\xi/2$, we compute for each $\psi\in\Psi$ a relative $\zeta$-approximation $\hat Z^\psi_{G}$ to $\tilde Z^\psi_{G}(\beta)$ via the cluster expansion and the method described in previous subsections.
This takes time
\[
O\left( q^{\ell} \cdot n (n/\zeta)^{O(\log \Delta)}\right)\,.
\]
Let $\hat Z_{G}=\sum_{\psi\in\Psi} \hat Z^\psi_{G}$ be the sum of these approximations. Then we have for each $\psi\in\Psi$,
\[
e^{-\zeta} \hat Z_{G}^\psi \le \tilde Z^\psi_G(\beta) \le e^{\zeta} \hat Z_{G}^\psi\,,
\]
and so summing over $\psi\in\Psi$ gives
\[
e^{-\zeta} \hat Z_{G'} \le \sum_{\psi\in\Psi}\tilde Z^\psi_G(\beta) \le  e^{\zeta}\hat Z_{G'}\,.
\]
But we also have
\[
e^{-e^{-n}} \tilde Z^\psi_{G}(\beta) \le Z^\psi_{G}(\beta) \le e^{e^{-n}} \tilde Z^\psi_{G}(\beta)\,,
\]
and so
\[
e^{-(\zeta+e^{-n})} \hat Z_{G} \le \sum_{\psi\in\Psi} Z^\psi_{G}(\beta) = Z^*_{G}(\beta) \le  e^{\zeta+e^{-n}}\hat Z_{G}\,.
\]
Finally, we have
\[
e^{-e^{-n}}  Z^*_{G'}(\beta) \le Z_{G}(\beta) \le e^{e^{-n}} Z^*_{G}(\beta)\,,
\]
and so
\[
e^{-(\zeta+2e^{-n})}\hat Z_{G} \le Z_{G}(\beta) \le  e^{\zeta+2e^{-n}}\Zalg\,.
\]
Then since $\xi > e^{-n/2}$ we have that $\hat Z_{G}$ is a relative $(\zeta+2\xi^2)$-approximation to $Z_{G}(\beta)$.
Without loss of generality let $\xi \le 1/4$ so that $\zeta=\xi/2$ satisfies $\zeta\le \xi-2\xi^2$, and we have a relative $\xi$-approximation as required.
To complete the proof of Theorem~\ref{thm:nobadparts}, we reiterate that the running time of this algorithm, assuming we are supplied with the required partition $P_1,\dotsc,P_{\ell}$, is
\[
O\left( q^{\ell} n \cdot (2n/\xi)^{O(\log \Delta)}\right)\,.
\]
An upper bound on the running time in either case $\xi\le e^{-n/2}$ and $\xi > e^{-n/2}$ is then
\[T_\lrg =  O\big(q^\ell\Delta n(2n/\xi)^{O(\log(q\Delta))}\big) \,.\]
This concludes the proof of \cref{thm:nobadparts}.

\subsection{Proof of Theorem~\ref{thm:Zalg}}

With the results collected in the previous subsections we can make precise the sketch of the proof given at the start of Section~\ref{sec:cluster}.
Suppose that the graph $G$ satisfies the assumptions of Theorem~\ref{thm:Zalg}, which means $G$ is on $n$ vertices, has maximum degree $\Delta$ and minimum degree $\delta$, and is supplied with a $(\phiin,\phiout,\tau\delta)$-partition $P_1,\dotsc,P_\ell$ such that $\ell<k$ and for some $\eta>0$ and $0\le s\le \ell$ we have $|P_i|<\eta n$ for all $1\le i\le s$ and $|P_i|\ge \eta n$ for $s+1\le i\le\ell$. 
We are given $\xi>0$ in the statement of Theorem~\ref{thm:Zalg}.

By Theorem~\ref{thm:JKP} we can compute a relative $\xi$-approximation $Z_i$ to each $Z_{G[P_i]}(\beta)$ for $1\le i\le s$ in time 
\[ T_\sml = O\big(s \Delta \eta n(2\eta n/\xi)^{O(\log(q\Delta))}\big)\,,\]
provided $\beta \ge (4+2\log(q\Delta))/(\phiin\tau\delta)$. 
Let $V' = \bigcup_{i=s+1}^\ell P_i$ and $G'=G[V']$, and observe that $G'$ has maximum degree $\Delta$ and $P_{s+1},\dotsc,P_\ell$ is a $(\phiin, \phiout, \tau\delta)$-partition of $V'$ into $\ell-s$ parts such that each part has size at least $\eta'|V'|$ with $\eta'= \eta n/|V'|\ge \eta$. 
Note that the minimum degree of $G'$ can be less than the minimum degree of $G$, but since we did not remove any edges from inside a good part, $G'[P_i]$ inherits minimum degree $\tau\delta$.
Then by Theorem~\ref{thm:nobadparts} we can compute a relative $\xi$-approximation to $Z_{G'}(\beta)$ in time 
\[T_\lrg =  O\big(q^{\ell-s}\Delta n(2n/\xi)^{O(\log(q\Delta))}\big) \,,\]
provided $\beta \ge (4+2\log(q\Delta))/(\phiin\tau\delta\eta')$, which is implied by $\beta \ge (4+2\log(q\Delta))/(\phiin\tau\delta\eta)$.

Let $G''$ be the disjoint union of $G'$ and $G[P_i]$ for $1\le i\le s$. 
Then $G''$ is obtained from $G$ by removing all edge sets $\partial(P_i)$ for $1\le i\le s$, which means removing $X \le s \phiout\Delta\eta n$ edges in total. 
This means that for any coloring $\omega : V\to [q]$ we have 
\[ m_{G}(\omega) \ge m_{G''}(\omega) \ge m_{G}(\omega) - X\,,\]
and so
\[ e^{-\beta X} Z_G(\beta) \le Z_{G''}(\beta)\le Z_G(\beta)\,, \]
meaning $e^{\beta X/2}Z_{G''}(\beta)$ is a relative $\beta X/2$-approximation to $Z_G(\beta)$. 
But we also have $Z_{G''}(\beta) = Z_{G'}(\beta)\cdot \prod_{i=1}^sZ_{G[P_i]}(\beta)$ and so via the separate relative $\xi$-approximations to each term of the product computed above we have that $e^{\beta X/2}\hat Z' \cdot \prod_{i=1}^s \hat Z_i$ is a relative $(s+1)\xi$-approximation to $e^{\beta X/2}Z_{G''}(\beta)$ and hence a relative approximation to $Z_G(\beta)$ with accuracy
\[ (s+1)\xi + \beta X/2 \le (s+1)\xi + \beta s\phiout\Delta\eta n/2\,, \]
as required accuracy for Theorem~\ref{thm:Zalg}.
To complete the proof it suffices to observe that 
\begin{align*}
T_\sml + T_\lrg &= O\big(s \Delta \eta n(2\eta n/\xi)^{O(\log(q\Delta))}\big) + O\big(q^{\ell-s}\Delta n(2n/\xi)^{O(\log(q\Delta))}\big) 
\\&= O\left(\left(s\eta^{O(\log(q\Delta))} + q^{\ell-s}\right)\Delta n\left(\frac{2n}{\xi}\right)^{O(\log(q\Delta))}\right) \,.
\end{align*}

\section{Spectral partitioning with a minimum degree condition}\label{sec:partitioning}

In this section we extend the method of~\cite{OT14} to give a minimum degree condition in the partitions found algorithmically, and the necessary modifications are fairly straightforward. 
This approach relies heavily on a result of Lee, Oveis Gharan, and Trevisan~\cite{LOT12} that generalizes the well-known Cheeger inequality to higher eigenvalues. 
We write 
\[
\rho_G(k) := \min_{\substack{A_1,\dotsc,A_k\\\text{disjoint}}} \max\{ \phi_G(A_i) : i\in[k]\}\,,
\]
for the \emph{$k$-way expansion} of a graph $G$, and note that the main result of~\cite{LOT12} is the following higher-order Cheeger inequality.
\begin{theorem}[Lee, Oveis Gharan and Trevisan~\cite{LOT12}]\label{thm:LOT12}
	For any graph $G$ and $k\ge 2$,
	\[ \frac{\lam_k}{2} \le \rho_G(k)\le O(k^2)\sqrt{\lam_k}\,. \]
\end{theorem}
This result implies a partition of $V$ into $P_1,\dotsc,P_\ell$ such that $\phi(P_\ell)$ is bounded above in terms of $\ell$ and $\lam_\ell$, but we are interested in extra properties of the induced subgraphs $G[P_i]$ related to expansion and minimum degree. 
The main result of~\cite{OT14} shows that one can obtain a partition as above with a lower bound on $\phi(G[P_i])$ controlled by $\lam_k$, and we develop a strengthening that adds a minimum degree condition.

Throughout this section we assume that $G=(V,E)$ is a graph on $n$ vertices.
We write $w(u,v) = 1$ if $uv$ is an edge of $G$, and $w(u,v)=0$ otherwise. 
This notation carries over from~\cite{OT14} where they work with the extra generality of edge-weighted graphs. 
Here we restrict our attention to usual graphs, but keep the notation of~\cite{OT14} for easy comparison.
For convenience, we assume that $G$ has no isolated vertices, which can easily be verified in $O(n)$ time. 
This prevents $\phi(S)$ being undefined when $S$ is a non-empty set of degree-zero vertices.
The interesting case for these results is a graph with no isolated vertices, and many authors tacitly assume this fact.
Before we give the algorithm, we collect some necessary results from~\cite{OT14}. 

\begin{lemma}[{Oveis Gharan and Trevisan~\cite[Lemma 1.13]{OT14}, see also~\cite{KLL+13,LOT12}}]\label{lem:OT14lam2Pi}
	There is a universal constant $C>0$ such that for any integer $k\ge 2$ and any partitioning of $V$ into $\ell< k$ pieces $P_1,\dotsc,P_\ell$ we have 
	\[
	\min_{i\in[\ell]}\lam_2(G[P_i]) \le 2Ck^6\lam_k(G)\,.
	\]
\end{lemma}

For $S, T\subset V$ we define 
\[ \e{S}{T} := \sum_{u\in S,v\in T-S} w(u,v)\,, \]
noting that $\e{S}{T}=\e{T}{S}$ when $T\cap S=\emptyset$, but in general the terms $
\e{S}{T}$ and $\e{T}{S}$ are not equal.
We also have for $S\subset V$ that 
\[ |\partial(S)| = \e{S}{V\setminus S} = \e{S}{V}\,. \]
For $S\subset B_i\subset V$ we define 
\[ \phi(S, B_i) := \frac{\e{S}{B_i}}{\frac{\vol(B_i\setminus S)}{\vol(B_i)}\cdot \e{S}{V\setminus B_i}}\,. \]

\begin{lemma}[{Oveis Gharan and Trevisan~\cite[Lemma 2.2]{OT14}}]\label{lem:OT14minphi}
	For any sets $S\subsetneq B\subset V$, if $\phi(S,B)\le \eps/3$ and 
	\[ \max\{\phi(S), \phi(B\setminus S)\} \ge (1+\eps)\phi(B)\,, \]
	then $\min\{\phi(S),\phi(B\setminus S)\}\le\phi(B)$. 
\end{lemma}

In our modification of the method of~\cite{OT14} we use a sharper version of this result for singleton sets $S$ given as the following claim.

\begin{claim}\label{claim:removevertex}
	Consider a vertex $u$ and subset $B$ such that $u\in B\subset V$, and write $d_V:=\vol_G(u)$, $d_B:=\vol_{G[B]}(u)$. 
	If $\vol(B)>d_V$ (i.e.\ if $B\setminus\{u\}$ contains a vertex with positive degree) then $\phi(B-u)\le \phi(B)$ if and only if $d_B \le (1-\phi(B))d_V/2$.
\end{claim}
\begin{proof}
	This follows easily from the definitions of $\phi$ and $\vol$ which give
	\[ \phi(B-u) = \frac{\vol(B)}{\vol(B)-d_V}\phi(B) - \frac{d_V-2d_B}{\vol(B)-d_V}\,, \]
	from which the result is immediate.
\end{proof}

\begin{lemma}[{Oveis Gharan and Trevisan~\cite[Lemma 2.3]{OT14}}]\label{lem:OT14phiGBiS}
	Let $S_B\subset B_i\subset V$ and write $\overline{S}_B=\overline{S}\cap B_i$. 
	If we have  $\vol(S_B)\le\vol(B_i)/2$ and 
	\[ \min\{\phi(S_B,B_i),\,\phi(\overline{S}_B,B_i)\}\ge\eps/3 \] 
	for some $0\le\eps\le1$, then 
	\[ \phi_{G[B_i]}(S_B) \ge \frac{\e{S_B}{B_i}}{\vol(S_B)} \ge \frac{\eps}{7}\max\{\phi(S_B),\,\phi(\overline{S}_B)\}\,. \]
\end{lemma}

Note that in the above result we are careful to take $\phi$ and $\vol$ without subscripts in the graph $G$, and where we are interested in $G[B_i]$ we use a subscript. 

\begin{lemma}[{Oveis Gharan and Trevisan~\cite[Lemma 2.5]{OT14}}]\label{lem:OT14phiPi}
	Suppose that $P_1,\dotsc,P_\ell$ form a partition of $V$, and for each $i\in[\ell]$ we have $B_i\subset P_i$ such that $\e{P_i\setminus B_i}{P_i} \ge \e{P_i\setminus B_i}{V}/\ell$. 
	Then for all $i\in[\ell]$ we have $\phi(P_i)\le \ell\phi(B_i)$.
\end{lemma}

\begin{proof}
	Write $S=P_i\setminus B_i$ and  note that because $B_i\subset P_i$ the assumption gives 
	\begin{equation}\label{eq:bestyolk} 
	\e{S}{P_i} = \e{S}{B_i}\ge \e{S}{V}/\ell\,.
	\end{equation}
	Then 
	\begin{align*} 
	\e{P_i}{V} 
	&= \e{B_i}{V} - \e{B_i}{S}  + \e{S}{V\setminus P_i}
	\\&\le \e{B_i}{V} - \e{S}{B_i}  + \e{S}{V}
	\\&\le \e{B_i}{V} + (\ell-1)\e{S}{B_i}\,,
	\end{align*}
	where the first inequality is because $S\cap B_i=\emptyset$ and $\e{S}{V\setminus P_i}\le \e{S}{V}$, and the second inequality follows from~\eqref{eq:bestyolk}. 
	This gives
	\begin{align*}
	\phi(P_i) = \frac{\e{P_i}{V}}{\vol(P_i)} 
	&\le \frac{\e{B_i}{V} + (\ell-1)\e{S}{B_i}}{\vol(B_i)}
	\\&=\phi(B_i) + \frac{(\ell-1)\e{B_i}{S}}{\vol(B_i)} \le \ell\phi(B_i)\,,
	\end{align*}
	where we use $S\cap B_i=\emptyset$ again, and $\e{B_i}{S}\le\e{B_i}{V}$.
\end{proof}

In the following lemma we have $B_i\subset P_i\subset V$ and $S\subset P_i$ with the notation 
\begin{align*}
S_B&:=B_i \cap S\,,&  \overline{S}_B &:= B_i \cap \overline{S}\,,\\
S_P&:=S \setminus B_i\,,& \overline{S}_P&:=\overline{S}\setminus B_i\,.
\end{align*}

\begin{lemma}[{Oveis Gharan and Trevisan~\cite[Lemma 2.6]{OT14}}]\label{lem:OT14phiS}
	Suppose that $B_i\subset P_i \subset V$, and let $S\subset P_i$ be such that $\vol(S_B)\le \vol(B_i)/2$. 
	Let $\xi$ be given such that $\xi\le \phi(S_P)$ and $\xi\le\max\{\phi(S_B),\,\phi(\overline{S}_B)\}$. 
	Suppose also that the following conditions hold for some $\ell\ge 1$ and $0<\eps<1$,
	\begin{enumerate}
		\item\label{itm:SP} If $S_P\ne\emptyset$ then $\e{S_P}{P_i}\ge \e{S_P}{V}/\ell$,
		\item\label{itm:SB} If $S_B\ne\emptyset$ then $\min\{\phi(S_B,B_i),\,\phi(\overline{S}_B, B_i)\} \ge \eps/3$.
	\end{enumerate}
	Then 
	\[ \phi(S)\ge \phi_{G[P_i]}(S) \ge \frac{\eps\xi}{14\ell}\,. \]
\end{lemma}
\begin{proof}
	Simply because $S\subset P_i$ we have
	\begin{align*}
		\phi(S) &= \frac{\e{S}{V}}{\vol(S)} = \frac{\e{S}{P_i} + \e{S}{V\setminus P_i}}{\vol_{G[P_i]}(S) +  \e{S}{V\setminus P_i}} \\&\ge \frac{\e{S}{P_i}}{\vol_{G[P_i]}(S)} = \phi_{G[P_i]}(S)\,,
	\end{align*}
	where we use that $(a+x)/(b+x) \ge a/b$ when $b\ge a \ge 0$, $b>0$ and $x\ge 0$. 
	Then it suffices to lower bound $\phi_{G[P_i]}(S)$. 
	
	First suppose that $\vol(S_B)\ge \vol(S_P)$, and hence $\vol(S)\le 2\vol(S_B)$. 
	We then have 
	\[ \phi_{G[P_i]}(S)\ge \frac{\e{S}{P_i}}{\vol(S)} \ge \frac{\e{S_B}{B_i}}{2\vol(S_B)}\,. \]
	Because of assumption~\ref{itm:SB} and $\vol(S_B)\le\vol(B_i)/2$ we may apply Lemma~\ref{lem:OT14phiGBiS} to obtain 
	\[ \frac{\e{S_B}{B_i}}{\vol(S_B)} \ge \frac{\eps}{7}\max\{\phi(S_B),\,\phi(\overline{S}_B)\} \ge \frac{\eps\xi}{7}\,, \]
	and hence $\phi_{G[P_i]}(S)\ge \eps\xi/14$, which is stronger than required. 
	
	Suppose instead that $\vol(S_P)>\vol(S_B)$. 
	Now we use the assumptions $\phi(S_B,B_i)\ge\eps/3$ and $\vol(S_B)\le\vol(B_i)/2$ to obtain 
	\begin{equation}\label{eq:eSBBlb}
	\e{S_B}{B_i} \ge \frac{\vol(\overline{S}_B)}{\vol(B_i)}\e{S_B}{V\setminus B_i}\frac{\eps}{3} \ge \e{S_B}{S_P}\frac{\eps}{6}\,. 
	\end{equation}
	Then we have 
	\begin{align*}
	\phi_{G[P_i]}(S) \ge \frac{\e{S}{P_i}}{\vol(S)} &\ge \frac{\e{S_P}{P_i\setminus S} + \e{S_B}{B_i}}{2\vol(S_P)}
	\\&\ge \frac{\e{S_P}{P_i\setminus S} + \e{S_B}{S_P}\eps/6}{2\vol(S_P)}\,.
	\end{align*}
	Using the observations that
	\begin{align*}
		\e{S_P}{P_i} &= \e{S_P}{P_i\setminus S} + \e{S_P}{S_B} 
		\\&= \e{S_P}{P_i\setminus S} + \e{S_B}{S_P}\,,
	\end{align*}
	and $\eps/6\le 1$ we continue from above, 
	\[ \phi_{G[P_i]}(S) \ge \frac{\eps}{12}\frac{\e{S_P}{P_i}}{\vol(S_P)} \ge \frac{\eps}{12\ell}\frac{\e{S_P}{V}}{\vol(S_P)} = \frac{\eps}{12\ell}\phi(S_P) \ge \frac{\eps\xi}{12\ell}\,, \]
	using the assumptions~\ref{itm:SP} and $\phi(S_P)\ge\xi$ in turn. 
\end{proof}

We also rely on the well-known \SP{} algorithm which efficiently finds a set of close to maximal conductance, see e.g.~\cite{OT14}.

\begin{theorem}\label{thm:SP}
	There is a near-linear time algorithm \SP{} that, given a graph $G=(V,E)$ finds a set $S\subset V$ such that $\vol(S)\le\vol(V)/2$ and $\phi(S)\le 4\sqrt{\phi(G)}$.
\end{theorem}

We now give a simple modification of the algorithm from~\cite{OT14} and prove Theorem~\ref{thm:partitioning}. 
We add some simple steps that ensure the required minimum degree conditions.
Let
\[ \rho^* := \min\left\{\lam_k/10, 30Ck^5\sqrt{\lam_{k-1}}\right\}\,, \]
where $C$ is the constant from Lemma~\ref{lem:OT14lam2Pi}, and write
\[ \phiin := \frac{\lam_k}{140k^2}\,,\qquad\phiout := 90Ck^6\sqrt{\lam_{k-1}}\,,\qquad \tau = \frac{1}{5(k-1)}\,.\]

\begin{algorithm}[htp]
	\begin{algorithmic}[1]
		\Input $k>1$ such that $\lambda_k>0$.
		\Output Sets $P_1,\dotsc,P_\ell$ that form a $(\phiin^2/4,\phiout, \tau)$ partitioning of $G$ for some $1\le \ell<k$.
		\State Let  $\ell=1$, $P_1=B_1=V$.
		\While {$\exists\ i\in\ell$ such that $\e{P_i\setminus B_i}{B_i} < \e{P_i\setminus B_i}{P_j}$ for $j\neq i$, or Spectral Partitioning finds $S\subseteq P_i$ such that
			$\max\{\phi_{G[P_i]}(S),\, \phi_{G[P_i]}(P_i\setminus S)\} < \phiin$\label{step:partitionwhile}} 
		\State Assume (after renaming) that $\vol(S\cap B_i) \leq \vol(B_i)/2$.  
		\State Let $S_B = S\cap B_i, {\overline{S}}_B = B_i \cap \overline{S}, S_P=(P_i\setminus B_i)\cap S$ and 
		${\overline{S}}_P = (P_i\setminus B_i)\cap \overline{S}$. \If { $\max\{\phi(S_B), \phi(\overline{S}_B)\} \leq \rho^* (1+1/k)^{\ell+1}$}\label{step:ifmaxSB}
		\State Let $B_i=S_B, P_{\ell+1}=B_{\ell+1}=\overline{S}_B$ and $P_i=P_i\setminus \overline{S}_B$. Set $\ell\leftarrow \ell+1$
		and \Goto{} step~\ref{step:vertexfix}.\label{step:addBlB}
		\EndIf
		\If { $\max\{\phi(S_B, B_i), \phi(\overline{S}_B,B_i)\} \leq 1/(3k)$,\label{step:ifchangeTi}}
		\State Update $B_i$ to either of $S_B$ or $\overline{S}_B$ with the smallest conductance, and \Goto{}
		step~\ref{step:vertexfix}.\label{step:changeTi}
		\EndIf
		\If {$\phi(S_P) \le \rho^* (1+1/k)^{\ell+1}$}\label{step:ifmaxSP}
		\State Let $P_{\ell+1}=B_{\ell+1}=S_P$, and update $P_i=P_i\setminus S_P$. Set $\ell\leftarrow \ell+1$ and \Goto{} step~\ref{step:vertexfix}.\label{step:addBlP}
		\EndIf
		\If { $\e{P_i\setminus B_i}{P_i} < \e{P_i\setminus B_i}{B_j}$ for $j\neq i$\label{step:ifPimTi} }
		\State Update  $P_j=P_j\cup (P_i\setminus B_i)$, and $P_i=B_i$ and \Goto{} step~\ref{step:vertexfix}.
		\EndIf
\If { $\e{S_P}{P_i} < \e{S_P}{P_j}$ for $j\neq i$,\label{step:ifPSPi}}
		\State Update $P_i = P_i-S_P$ and merge $S_P$ with $\argmax_{P_j}\e{S_P}{P_j}$. 
		\EndIf
		
		\While {$\exists\ i\in[\ell]$ and $v\in B_i$ such that $\vol_{G[B_i]}(v) < \vol_G(v)/5$}\label{step:vertexfix}\label{step:ifmindegB}
		\State Update $B_i  = B_i \setminus\{v\}$.\label{step:mindegB}
		\EndWhile
		\While {$\exists\ i\in[\ell]$ and $v\in P_i\setminus B_i$ such that $\e{v}{P_i} < \e{v}{P_j}$ for $j\neq i$}\label{step:ifmindegP}
		\State Update $P_i = P_i \setminus \{v\}$ and insert $v$ into $\argmax_{P_j}\e{v}{P_j}$\,.\label{step:mindegP}
		\EndWhile\label{step:endfix}
		
		\EndWhile
		\Return $P_1,\dotsc,P_\ell$.
	\end{algorithmic}
	\caption{A polynomial time algorithm for partitioning $G$ into $k$ expanders.\label{alg}}
\end{algorithm}

We note that lines 1 to 19 of Algorithm~\ref{alg} are identical to~\cite[Algorithm 3]{OT14} except for the trivial modification that the \Goto{} statements jump to our new steps, the while loops at lines~\ref{step:vertexfix}--\ref{step:endfix}, instead of to line~\ref{step:partitionwhile}. 
That is, we have simply added some extra work to the end of each iteration of the main while loop in~\cite[Algorithm 3]{OT14}. 
This extra work moves vertices that have small degree inside the relevant sets to other sets. 
At the termination of the algorithm we have $\mu_{G[B_i]}(v) \ge \mu_G(v)/5$ for all $v\in B_i$, a strong bound independent of $\ell$, but for $v\in P_i\setminus B_i$ we will have $\mu_{G[P_i]}(v) \ge \mu_G(v)/\ell$.
For convenience, we work with the stated $\tau=1/(5(k-1))$ instead of $\tau = \min\{1/5,\,1/\ell\}$ that our proof actually gives.

During the entire run of the algorithm $B_1,\dotsc,B_\ell$ are disjoint, $B_i\subset P_i$ for all $i\in[\ell]$, and $P_1,\dotsc, P_\ell$ form a partitioning of $V$. 
We will prove that the algorithm terminates with $\ell<k$, but the following claim contains $
\ell<k$ as a hypothesis for convenience.
Once we have proved that $\ell<k$ throughout, this assumption is seen to be automatically satisfied. 

\begin{claim}\label{claim:phiBi}
	Throughout the execution of the algorithm, provided $\ell<k$ we have
	\[ \max_{i\in[\ell]}\{ \phi(B_i) \} \le \rho^*(1+1/k)^\ell\,. \]
\end{claim}
\begin{proof}
	At the start we have $\ell=1$ and $B_i=V$ so $\phi(B_i)=0$. 
	We prove the claim inductively from here, noting that the only steps in which some $B_i$ is modified are lines \ref{step:addBlB}, \ref{step:changeTi}, \ref{step:addBlP}, and \ref{step:mindegB}. 
	
	Both steps \ref{step:addBlB} and \ref{step:addBlP} are designed specifically to maintain the induction hypothesis;
	before incrementing $\ell$ on those lines we obtain $\ell+1$ sets such that 
	\[ \max_{i\in[\ell]}\{ \phi(B_i) \} \le \rho^*(1+1/k)^{\ell+1}\,, \]
	as required.
	
	If step~\ref{step:changeTi} is executed then the condition on line~\ref{step:ifmaxSB} is not satisfied, giving 
	\[ \max\{\phi(S_B), \phi(\overline{S}_B)\} > (1+1/k)^{\ell+1} \rho^* \ge (1+1/k)\phi(B_i) \]
	by the induction hypothesis. 
	But if step~\ref{step:changeTi} is executed we must also satisfy the condition on line~\ref{step:ifchangeTi}, giving 
	\[ \max\{\phi(S_B, B_i), \phi(\overline{S}_B,B_i)\} \leq \frac{1}{3k}\,. \]
	Then by Lemma~\ref{lem:OT14minphi} with $\eps=1/k$ and $S=S_B$ and the induction hypothesis we have 
	\[ \min\{\phi(S_B),\phi(\overline{S}_B)\} \le \phi(B_i) \le \rho^*(1+1/k)^\ell\,,\]
	so making whichever of $S_B$ and $\overline{S}_B$ has smaller conductance the new $B_{\ell+1}$ satisfies the required bound.
	
	The above arguments are exactly as in~\cite{OT14}, and for our modification we have to analyze step~\ref{step:mindegB}.
	By the induction hypothesis $\phi(B_i)\le\rho^*(1+1/k)^\ell$ it suffices to show that whenever $v\in B_i$ has $\vol_{G[B_i]}(v) < \vol_G(v)/5$, we have $\phi(B_i\setminus\{v\})\le\phi(B_i)$. 
	To see this, note that the induction hypothesis, the assumption that $\ell<k$, and the definition of $\rho^*$ (in which we use $\lam_k\le 2$) together mean that $\phi(B_i)\le 2e/10$, and hence that $|B_i|\ge 2$. 
	This is because any singleton set (containing a vertex of positive degree) has conductance $1$.
	So we may apply Claim~\ref{claim:removevertex}, and hence it suffices to show that 
	\[ \vol_{G[B_i]}(v) \le (1-\phi(B_i))\vol_G(v)/2\,. \]
	But by the induction hypothesis we have $(1-\phi(B_i))/2 \ge (1-2e/10)/2 \ge 1/5$, so the condition in step~\ref{step:ifmindegB} gives the required bound.
	We note for use below that we only needed the assumption $\ell<k$ to analyze step~\ref{step:mindegB}.
\end{proof}

\begin{claim}\label{claim:ellk}
	During the execution of the algorithm we always have $\ell <k$. 
\end{claim}
\begin{proof}
	We start with $\ell=1$ and $\ell$ is only ever incremented by $1$ at a time, so it suffices to show that a step which increments $\ell$ will never be executed when $\ell = k-1$. 
	The relevant steps are in lines~\ref{step:addBlB} and~\ref{step:addBlP}, and supposing that one of these steps causes $\ell$ to be incremented to $k$ we analyze the sets $B_1,\dotsc,B_k$ that exist immediately after this step, and hence \emph{before} the jump \Goto{} step~\ref{step:vertexfix} completes.
	By the proof of the above claim, in which $\ell<k$ was only needed for the step~\ref{step:mindegB}, any such $B_1,\dotsc,B_k$ must have $\phi(B_i) \le \rho^*(1+1/k)^{\ell} < e\rho^*$ for all $i\in[\ell]$.
	But then we have disjoint sets $B_1,\dotsc,B_k$ such that
	\[ \max_{i\in[\ell]}\{\phi(B_i)\} < e\rho^* < \lam_k /2\, \]
	by the definition of $\rho^*$. 
	But this implies $\rho(k)<\lam_k/2$, contradicting Theorem~\ref{thm:LOT12}.
\end{proof}

Before proving that our minimum degree condition holds, we note that for $v\in U\subset V$ we have by definition that $\e{v}{U}=\mu_{G[U]}(v)$.

\begin{claim}
	Whenever the algorithm checks the condition in line~\ref{step:partitionwhile} to determine if execution continues, we have for all $i\in[\ell]$ and $v\in P_i$ that $\vol_{G[P_i]}(v) \ge \tau\vol_G(v)$.
\end{claim}
\begin{proof}
	This is a simple consequence of the loops at lines~\ref{step:vertexfix}--\ref{step:endfix}. 
	They ensure that for all $i\in[\ell]$ and $v\in B_i$ we have the stronger condition $\vol_{G[B_i]}(v) \ge \vol_G(v)/5$. 
	We also have that any $v\in V\setminus \bigcup_{i\in[\ell]}B_i$ is in $P_i$ such that for all $j\in[\ell]$.
	\[ \e{v}{P_i} \ge \e{v}{P_j}\,. \]
	But since the sets $P_1,\dotsc,P_\ell$ form a partition of $V$ we have 
	\[ \vol_G(v) = \sum_{j\in[\ell]} \e{v}{P_j}\,, \]
	and hence $\vol_{G[P_i]}(v) = \e{v}{P_i} \ge \vol_{G}(v)/\ell$. 
	With the facts that $\ell<k$ and $\tau=1/(5(k-1))$ we have the required degree conditions.
\end{proof}

\begin{claim}
	If the algorithm terminates then the sets $P_1,\dotsc,P_\ell$ form a $(\phiin^2/4,\phiout,\tau \delta)$-partition of $V$. 
\end{claim}
\begin{proof}
	Suppose that the algorithm terminates with some $\ell<k$ and sets $B_1\dotsc,B_\ell$ and $P_1,\dotsc,P_\ell$. 
	By the above claim we have the required degree conditions, and hence it suffices to show that $\phi(G[P_i])\ge \phiin^2/2$ and $\phi(P_i) \le \phiout$ for all $i\in[\ell]$. 
	
	By the condition in line~\ref{step:partitionwhile} and the performance of the \SP{} algorithm we have $\phi(G[P_i])\ge \phiin^2/4$ as required.
	
	Moreover, by the same condition in line~\ref{step:partitionwhile} we have for each $i\in[\ell]$ that 
	\[ \e{P_i\setminus B_i}{B_i} \ge \e{P_i\setminus B_i}{V}/\ell\,,\]
	and by Lemma~\ref{lem:OT14phiPi}, Claim~\ref{claim:phiBi} and the fact that $\ell<k$ we have
	\[ \phi(P_i) \le \ell\phi(B_i) \le \ell e \rho^* \le 90Ck^6\sqrt{\lam_{k-1}}\,,\]
	as required. 
\end{proof}

Then it remains to show that the algorithm terminates and bound its running time. 

\begin{claim}
	In each iteration of the main loop starting at line~\ref{step:partitionwhile}, at least one of the conditions in lines~\ref{step:ifmaxSB}, \ref{step:ifchangeTi}, \ref{step:ifmaxSP}, and \ref{step:ifPimTi} holds.
\end{claim}
\begin{proof}
	We use Lemma~\ref{lem:OT14phiS} to show that if none of the conditions in lines~\ref{step:ifmaxSB}, \ref{step:ifchangeTi}, \ref{step:ifmaxSP}, and \ref{step:ifPimTi} holds then $\phi_{G[P_i]}(S)\ge \phiin$, which is a contradiction. 
	Then we suppose that none of the conditions hold. 
	
	Since the conditions in lines~\ref{step:ifchangeTi} and~\ref{step:ifPSPi} do not hold, assumptions 1\ and 2\ of Lemma~\ref{lem:OT14phiS} are satisfied with $\eps=1/k$. 
	And since condition in lines~\ref{step:ifmaxSB} and~\ref{step:ifmaxSP} do not hold we have the following facts which use the definition of $\rho(\ell+1)$ as the $(\ell+1)$-way expansion of $G$, and Claim~\ref{claim:phiBi}:
	\begin{align*} 
	\max\{\phi(&S_B,\overline{S}_B)\}
	\\&= \max\{\phi(B_1),\dotsc,\phi(B_{i-1}),\phi(S_B),\phi(\overline{S}_B),\phi(B_{i+1}),\dotsc,\phi(B_\ell)\} \\&\geq \max\{\rho^*,\rho(\ell+1)\}\,,
	\end{align*}
	and
	\begin{align*}
	\phi(S_P) &= \max\{\phi(B_1),\dotsc,\phi(B_\ell), \phi(S_P)\} \geq \max\{\rho^*,\rho(\ell+1)\}\,. 
	\end{align*}
	So with $\rho =\rho^*$  and $\eps=1/k$, by Lemma~\ref{lem:OT14phiS} we have
	\begin{equation}\label{eq:rholowerbound}
	\phi(S) \geq \frac{\eps \rho}{14 k} = \frac{\max\{\rho^*,\rho(\ell+1)\}}{14k^2}\,,
	\end{equation}
	where we use Theorem~\ref{thm:LOT12} for the final equality.
	Now, if $\ell=k-1$, then by Theorem~\ref{thm:LOT12} we have
	\[ \phi(S) \geq \frac{\rho(k)}{14k^2} \geq \frac{\lambda_k}{28k^2} \geq \phiin\,,\]
	which is a contradiction, so we are done.
	Otherwise, we must have $\ell<k-1$ and so by Lemma~\ref{lem:OT14minphi}, 
	\begin{equation}
	\label{eq:phisupperbound}
	\phi(S) \leq \min_{i\in[\ell]} \sqrt{2\lambda_2(G[P_i])} \leq \sqrt{4C k^6\lambda_{k-1}}, 
	\end{equation}
	where the first inequality follows from Cheeger's inequality $\phi(G)\le\sqrt{2\lam_2}$.
	Putting \eqref{eq:rholowerbound} and \eqref{eq:phisupperbound} together we have
	\[ \rho^* \leq 14k^2\sqrt{4C k^6 \lambda_{k-1}}\,.\]
	But then by definition of $\rho^*$ we must have $\rho^*=\lambda_k/10$, so by \eqref{eq:rholowerbound} we have the desired contradiction
	\[ \phi(S) \geq \frac{\lambda_k}{140k^2} =\phiin\,.\qedhere\]
\end{proof}

We can easily bound the running time of this algorithm in terms of the running time of~\cite[Algorithm 3]{OT14} by bounding the amount of extra work in lines~\ref{step:vertexfix}--\ref{step:endfix} that we add to each iteration. 

\begin{claim}
	The algorithm terminates in at most $O(kn|E|)$ iterations of the main loop, and each iteration takes time at most $O(n|E|)$, yielding a running time of at most $O(kn^2|E|^2)$.
\end{claim}
\begin{proof}
	In each iteration of the main loop at least one of conditions in lines~\ref{step:ifmaxSB}, \ref{step:ifchangeTi}, \ref{step:ifmaxSP}, \ref{step:ifPimTi}, \ref{step:ifPSPi}
	is satisfied.
	By Claim~\ref{claim:ellk}, lines~\ref{step:ifmaxSB} and \ref{step:ifmaxSP} can be satisfied at most $k-1$ times combined. 
	Line~\ref{step:ifchangeTi} can be satisfied at most $n$
	times because each time the size of some $B_i$ decreases by at least one. 
	For a fixed $\ell$ and $B_1 ,\dotsc,B_\ell$ the conditions in lines \ref{step:ifPimTi} and \ref{step:ifPSPi} can hold at most $O(|E|)$ times combined because each time the number of edges between $P_1,\dotsc,P_\ell$ decreases by at least one.
	This shows that the main loop can execute at most $O(kn|E|)$ times, as in the proof of Oveis Gharan and Trevisan~\cite{OT14}.
	The additional work of the loop at line~\ref{step:ifmindegB} takes time $O(n)$ because the size of some $B_i$ decreases by one each iteration, and the additional work of the loop at line~\ref{step:ifmindegP} takes time $O(|E|)$ because the number of edges between the $P_i$ decreases by at least one each iteration. 
	This completes the proof of a running time bound of $O(kn^2|E|^2)$.
\end{proof}

\bibliographystyle{habbrv}
\bibliography{potts}

\end{document}